\documentclass[fleqn,12pt,twoside]{article}

\usepackage{amsmath} 
\usepackage{amsfonts} 
\usepackage{amsthm}
\usepackage{dcolumn}
\usepackage{bm}
\usepackage{chemarrow}
\usepackage{graphicx}
\usepackage{epsfig}
\usepackage{mathrsfs}
\usepackage{nicefrac}
\usepackage{cite}
\usepackage{hyperref}
\usepackage[all,2cell]{xy}

\newcommand{\be}{\begin{equation}}
\newcommand{\ee}{\end{equation}}
\newcommand{\bea}{\begin{eqnarray}}
\newcommand{\eea}{\end{eqnarray}}
\newcommand{\bes}{\begin{subequations}\bea}
\newcommand{\ees}{\eea\end{subequations}}
\newcommand{\ba}{\begin{array}}
\newcommand{\ea}{\end{array}}
\newcommand{\im}{\mathrm{im}\,}
\newcommand{\bra}[1]{\langle #1|}

\newcommand{\ket}[1]{|#1\rangle}
\newcommand{\braket}[2]{\langle #1|#2\rangle}

\newcommand{\hhat}[1] {#1}
\newcommand{\Id}{I}
\newcommand{\E}{|E|}
\newcommand{\V}{|V|}
\newcommand{\C}{|C|}

\begin{document}

\newtheorem{theorem}{Theorem}

\title{Cycle/cocycle oblique projections \\ on oriented graphs}

\author{Matteo Polettini\footnote{matteo.polettini@uni.lu} \\ {\small Complex Systems and Statistical Mechanics, University of Luxembourg,} \\ {\small 162a avenue de la Fa\"iencerie, L-1511 Luxembourg (G. D. Luxembourg)}}

\date{\vspace{-5ex}}





\maketitle

\begin{abstract}
It is well known that the edge vector space of an oriented graph can be decomposed in terms of cycles and cocycles (alias cuts, or bonds), and that a basis for the cycle and the cocycle spaces can be generated by adding and removing edges to an arbitrarily chosen spanning tree.  In this paper we show that the edge vector space can also be decomposed in terms of cycles and the generating edges of cocycles (called {\it cochords}), or of cocycles and the generating edges of cycles (called {\it chords}). From this observation follows a construction in terms of oblique complementary projection operators. We employ this algebraic construction to prove several properties of unweighted Kirchhoff-Symanzik matrices, encoding the mutual superposition between cycles and cocycles. In particular, we prove that dual matrices of planar graphs have the same spectrum (up to multiplicities). We briefly comment on how this construction provides a refined formalization of Kirchhoff's mesh analysis of electrical circuits, which has lately been applied to generic thermodynamic networks.
\end{abstract}

{\small
{\bf Keywords}: Cycles, cocycles, oblique projections

{\bf MSC (2010)}:  05C50, 81Q30
}



\newpage

\section{Introduction}

Kirchhoff's mesh analysis is a systematic method to solve for the stationary currents of an electrical circuit produced by current and voltage sources \cite{iyer}. The method is based on Kirchhoff's Current Law, which implements charge conservation at the nodes of the network, and Kirchhoff's Loop Law, which follows from Faraday's law for the circuitation of the electric field within a metallic conductor. An electro-motive force generates cycling currents in the network's ``meshes'' (hereby called {\it cycles}), while an input current streaming thorugh the network ``cuts'' the graph into a set of source nodes and a set of sink nodes. Hence the graph-algebraic analysis of cycles and cuts (hereby called {\it cocycles}) is tightly intertwined with the constraints that maintain a system out of equilibrium. The same analysis can be generalized to arbitrary networks carrying fluxes of any kind; in particular mesh analysis has become relevant to nonequilibrium thermodynamics, after the pioneering works of Hill \cite{hill} and Schnakenberg \cite{schnak} have identified network cycles with actual thermodynamic cycles that a (biochemical) machine performs.

In this paper we elaborate on the algebraic cycle/cocyle analysis of oriented graphs, introducing a technique based on complementary oblique projections. We exploit the formalism to prove several facts regarding the spectrum and the eigenvectors of dual unweighted Kirchhoff-Symanzik matrices, which encode the information about the superpositions of cycles among themselves, and of cocycles among themselves.

The paper is structured as follows. In Sec.\,\ref{setup} we set the stage for cycle and cocycle analysis of oriented graphs, in particular refreshing the construction of a basis for the cycle and cocyle spaces starting from a spanning tree, by addition and removal of generating edges. We introduce cycle and cocycle oblique projections $P,Q$ in \S\,\ref{res1} and prove that they are complementary. In \S\,\ref{res2} we introduce the 2-form $\Omega$ encoding mutual superpositions of cycles and cocycles with the generating edges. In \S\,\ref{res3} we prove several facts regarding the spectrum and eigenvectors of $P^\dagger P$ and $QQ^\dagger$, leading to the relationship $P^\dagger P + QQ^\dagger = I-\Omega^2$. In \S\,\ref{res4} we discuss the self-dual behavior of the construction under duality of planar graphs, and in \S\,\ref{res5} we discuss the relationship between KS and Laplacian spectra. We discuss the physical rationale in Sec.\,\ref{thermo} with reference to network thermodynamics, and draw conclusions in Sec.\,\ref{conclu}. Some not-well-known linear algebra facts that are relevant to the theory are rederived in Appendix\,\ref{appendix}. Step by step, all results are illustrated by an example in Appendix\,\ref{example}.

\section{\label{setup}Setup}

We consider an oriented connected graph $G = (V,E,\partial)$, with the $\V\times \E$ matrix $\partial$ prescribing the incidence relationships between $\E$ oriented edges $e\in E$ and $\V$ vertices $v \in V$. We remind that in its $(v,e)$ position, the incidence matrix has entry $+1$ if edge $e$ goes into vertex $v$, $-1$ if edge $e$ comes out of vertex $v$, and $0$ otherwise. The inverse edge is denoted $-e$.  As a matter of fact, invoking orientation is necessary for the sake of the algebraic treatment, but the results are ultimately independent of orientation.

We will be involved with subsets of edges, which will be represented as integer-valued vectors living in the vector space $\mathbb{R}^{\E}$ of real-valued linear combinations of edges. We employ the Dirac notation $\ket{\,\cdot\,}$ for vectors and $\bra{\,\cdot\,}$ for covectors. The action of a covector on a vector is denoted with $\braket{\,\cdot\,}{\,\cdot\,}$, $\ket{\,\cdot\,}\bra{\,\cdot\,}$ is the outer product, $\wedge $ the exterior product  and $\otimes$ the tensor product of vectors. We endow $\mathbb{R}^{\E}$ with the Euclidean scalar product, inducing an isomorphism between vectors and covectors, namely  $\ket{\,\cdot\,} = \bra{\,\cdot\,}^\dagger $, where $\dagger$ denotes transposition.

The {\it oriented cycle space} $\mathcal{C}$ is defined as the null space of $\partial$. Graphically, integer oriented cycles are sequences of edges such that for each incoming edge at a vertex there is an outgoing one. The {\it oriented cocycle space} $\ast\mathcal{C}$ is defined as the row space of $\partial$. The rows of $\partial$ are {\it simple} cocycles, that is, minimal sets of edges whose removal disconnects the graph into two subgraphs. The vertex set from which a simple cocycle emanates is called {\it source}, that to which the cocycle directs is called {\it sink}. The cycle and cocycle spaces are orthogonal subspaces spanning $\mathbb{R}^{\E}$. The dimension of the cocycle space is easily shown to be $\V-1$, hence by the rank-nullity theorem the dimension of the cycle space is the cyclomatic number $\C = \E-\V+1$. In the following, the cocycle index $\mu$ ranges from $1$ to  $\V-1$, the cycle index $\alpha$ ranges from $\V$ to $\E$, and the generic index $i = \mu,\alpha$ ranges from $1$ to $\E$.

There is a standard procedure to build a basis for the cycle and the cocycle spaces out of an arbitrary choice of a spanning tree $T \subseteq E$ \cite{nakanishi} (see Appendix \ref{example} for an example). Let us briefly recall it. A spanning tree $T$ satisfies two properties, both of which define it: It is a maximal subset of  edges that does not enclose a cycle; It is a minimal subset of  edges that connects all vertices. It is easily proven that the number of edges in a spanning tree is $\V-1$.  Edges $e_\mu \in T$ belonging to the spanning tree are called {\it cochords}. Edges $e_\alpha \notin T$ not belonging to the spanning tree are called {\it chords}.  Then, adding a chord to $T$ encloses a {\it cycle} $c_\alpha$. The cycle is consistently  oriented along the orientation of the generating chord. Removing a cochord from $T$ splits the graph into disconnected source and sink subsets, identifying a unique {\it cocycle} $c_\mu$ oriented along the direction of the generating cochord.

We let vectors $\ket{c_\alpha}$, $\ket{c_\mu}$, $\ket{e_\alpha}$, $\ket{e_\mu}$ denote the oriented edge composition respectively of cycles, cocycles, chords and cochords. Let the chord (resp. cochord) space $\mathcal{E}$ (resp.  $\ast\mathcal{E}$) be defined as the span of chords (resp. cochords). The chord space is orthogonal to the cochord space. Chords and cochords together form the {\it edge basis} serving as the natural orthonormal basis for the edge space\footnote{To avoid complications in the notation, we will systematically confuse operators with their matrix representation in the edge basis.}. It is also easily shown that the cycles generated from a spanning tree  $\ket{c_\alpha}$ span the cycle space $\mathcal{C}$, and similarly for the cocycle space $\ast\mathcal{C}$. Therefore, the choice of a spanning tree naturally induces two orthogonal decompositions of the edge vector space $\mathbb{R}^{\E} = \mathcal{E} \oplus \ast\mathcal{E} = \mathcal{C} \oplus \ast\mathcal{C}$.

\section{Results}

\subsection{\label{res1}Cycle/cocycle projections}

We first complement the above decompositions of the edge space of an oriented graph.

\begin{theorem}
The edge vector space can be decomposed in terms of cycles and cochords and in terms of cocycles and chords as $\mathbb{R}^{\E} = \mathcal{C} \oplus \ast\mathcal{E} = \ast\mathcal{C} \oplus \mathcal{E}$.
\end{theorem}

\begin{proof}
The crucial observation is that the above basis vectors satisfy the following orthogonality relations:
\be
\braket{e_\alpha}{e_\mu} = \braket{c_\alpha}{c_\mu} = 0, \quad \braket{e_\alpha}{c_{\alpha'}} = \delta_{\alpha,\alpha'}, \quad \braket{e_\mu}{c_{\mu'}} = \delta_{\mu,\mu'}. \label{eq:ortho}
\ee
The first two are trivial, the second two follow by construction, since a chord (resp. cochord) only belongs to the cycle (resp. cocycle) it generates. Then it is simple to prove that vectors $\{\ket{c_\alpha},\ket{e_\mu}\}$ are independent. In fact, cochords are linearly independent among themselves and with respect to chords. By the above algebraic relations, $\ket{c_\alpha}$ is the only cycle that is nonnull in the $\alpha$-th position. Then, cycles are independent of cochords and among themselves. Similarly, cocycles are independent of chords and among themselves. The conclusion follows.\end{proof}

Notice that such decompositions are not orthogonal, which preludes to our central result. Let us define the following operators on the edge space:
\bes
\hhat{P} & := & \sum_\alpha \ket{c_\alpha} \bra{e_\alpha}, \\
\hhat{Q} & := & \sum_\mu \ket{e_\mu} \bra{c_\mu}.
\ees

\begin{theorem}
Operators $\hhat{P}$ and $\hhat{Q}$ are oblique complementary projections.
\end{theorem}

\begin{proof}
By the orthogonality conditions Eq.\,(\ref{eq:ortho}), it is easily verified that $\hhat{P}\hhat{Q} = \hhat{Q}\hhat{P} = 0$, $\hhat{P}^2 = \hhat{P}$ and $\hhat{Q}^2 = \hhat{Q}$. Then $\hhat{P}$ are $\hhat{Q}$ are projections. But for a special subclass of graphs\footnote{\label{foot}The special class of graphs for which $\hhat{P}$ and $\hhat{Q}$ are orthogonal is such that every chord is a cycle and every cochord is a cocycle. Then, every chord is a loop (an edge that connects a vertex to itself) and every cochord belongs to a spanning tree, e.g.
\be
\xymatrix{ \bullet \ar@{-}@(dr,dl)[]  \ar@{-}@(ul,dl)[]  \ar@{-}[r] & \bullet  \ar@{-}[rd]  \ar@{-}[r] & \bullet  \\  &  & \bullet  \ar@{-}@(ur,dr)[] }
\ee
}, they are typically not orthogonal. Their kernels and images are given by
\bea
 \im \hhat{Q} = \ker \hhat{P} = \ast\mathcal{E}, \quad \ker \hhat{Q} = \im \hhat{P} = \mathcal{C} .
\eea
Hence, $P$ and $Q$ are complementary,
\be
\Id = \hhat{P} + \hhat{Q}, \label{eq:dec}
\ee
where $\Id$ is the identity on the edge space, $\Id = \sum_{i = \alpha,\mu} \ket{e_i} \bra{e_i}$. This relation can also be verified directly by applying $P+Q$ onto the basis $\{\ket{c_\alpha},\ket{e_\mu}\}$.
\end{proof}

\subsection{\label{res2}Mutual superpositions}

A straightforward corollary of the orthogonality relations in Eq.\,(\ref{eq:ortho}) is the known fact that if a cochord $e_\mu$ belongs to a cycle $c_\alpha$, then its conjugate cocycle $c_\mu$ contains the generating chord $e_\alpha$ in reverse direction \cite{biggs}.

\begin{theorem}\label{super}
The mutual projections of cycles onto cochords and of cocycles onto chords obey
\be
\braket{e_\mu}{c_\alpha} + \braket{c_\mu}{e_{\alpha}} = 0. \label{eq:anticom}
\ee
\end{theorem}
\begin{proof}
We decompose $\ket{e_\alpha}$ by means of Eq.~(\ref{eq:dec}), $\ket{e_\alpha} = \ket{c_\alpha}+ \sum_\mu \ket{e_\mu} \braket{c_\mu}{e_\alpha}$ and $\ket{c_\alpha}$ in the edge basis, $\ket{c_\alpha} = \ket{e_\alpha}  + \sum_\mu \ket{e_\mu} \braket{e_\mu}{c_\alpha}$.
The conclusion follows from the linear independence of vectors $\ket{e_\mu}$.
\end{proof}

The above superpositions are encoded in the cycle and the cocycle 2-forms, respectively defined by
\bes
\hhat{\Omega} := \sum_\alpha \bra{c_\alpha} \wedge \bra{e_\alpha} \\
\ast\hhat{\Omega} : = \sum_\mu \bra{c_\mu}  \wedge \bra{e_\mu}.
\ees
In fact, the action of $\Omega$ on $\ket{e_i}\otimes \ket{e_j}$ yields
\bea
\Omega \, \ket{e_i}\otimes \ket{e_j}  & = &  \sum_\alpha \Big( \braket{c_\alpha}{e_i}\braket{e_\alpha} {e_j} 
 -  \braket{e_\alpha} {e_i}  \braket{c_\alpha}{e_j}  \Big) \nonumber \\
 & = & \left\{ \ba{ll}
 0 , & \mathrm{if\,} (i,j) = (\alpha,\alpha')\\ 
 0 , & \mathrm{if\,}  (i,j) = (\mu,\mu') \\ 
- \braket{c_\alpha}{e_\mu} , &  \mathrm{if\,}  (i,j) = (\alpha,\mu) \\
+ \braket{c_\alpha}{e_\mu} , &  \mathrm{if\,} (i,j) = (\mu,\alpha) \\ 
 \ea \right. .
 \eea
\begin{theorem}
The cycle and the cocycle 2-forms coincide
\be
\hhat{\Omega} = \ast\hhat{\Omega}.
\ee
\end{theorem}
\begin{proof}
For a direct proof, one can evaluate the action of $\ast\Omega$ on $\ket{e_i}\otimes \ket{e_j}$ and the conclusion follows after Th.\,\ref{super}. More elegantly, rewriting
\be
\Omega - \ast\Omega = \sum_\alpha \bra{c_\alpha} \wedge \bra{e_\alpha} +  \sum_\mu \bra{e_\mu}  \wedge \bra{c_\mu} 
\ee
we notice that $\Omega - \ast\Omega$ is the sum of exterior products of dual basis covectors. It is a general fact that the latter vanishes.
\end{proof}

It is standard to define an operator $\Omega: \mathbb{R}^{\E} \to \mathbb{R}^{\E}$ by posing $\Omega \ket{w} :=  (\Omega (\,\cdot\, \otimes  \ket{w} ))^\dagger $. For sake of simplicity we do not distinguish between the bilinear form and its associated linear operator. As regards the latter, one easily obtains
\be
\Omega = P - P^\dagger = Q^\dagger - Q. \label{eq:ompq}
\ee
Notice that $\Omega = 0$ if and only if the projections are orthogonal, which occurs for the special class of graphs described in Footnote \ref{foot}.

Finally, for later use we define the {\it superposition matrix}  $\omega = (\Omega_{\mu\alpha})_{\mu,\alpha}$ which is the upper-right block of $\Omega$ in the edge basis, with entries
\be
\Omega_{\mu\alpha} = \left\{\ba{ll}
+1, & \mathrm{if\,chord\,} e_\alpha \mathrm{\,belongs\, to\, cocycle\,} -c_\mu  \\
-1, & \mathrm{if\,chord\,} e_\alpha \mathrm{\,belongs\, to\, cocycle\,} c_\mu  \\
0, & \mathrm{otherwise}. \ea \right.
\ee

\subsection{\label{res3}Kirchhoff-Symanzik matrices and their spectra}

We consider operators $P^\dagger P$ and $QQ^\dagger$, and in particular their eigenvalues, which by definition are the squared singular values of $P$ and $Q$ respectively.

Operator $P^\dagger P$ leaves the chord space $\mathcal{E}$ invariant and its kernel is the cochord space. Therefore we can define its restriction $K = P^\dagger P |\mathcal{E}$ to the chord space, which is invertible. The matrix elements of $K$ encode the mutual superpositions of cycles:
\begin{multline}
K_{\alpha \alpha'} = \braket{c_\alpha}{c_{\alpha'}} = \# (\textrm{shared edges with same orientation})\\  -  \# (\textrm{shared edges with opposite direction} ).
\end{multline}
In the edge basis, $K$ is the unweighted case of the so-called Kirchhoff-Symanzik (KS) matrix that features prominently in the theory of the parametric formulas for Feynman graphs \cite{nakanishi,bogner,marcolli}, where the weights are given by the propagators of a quantum field. Therefore, we will call it {\it cycle KS matrix}. Being a Gramian matrix, $K$ is symmetric positive-definite.

Similarly, let us restrict $QQ^\dagger$ to the cochord space, where it is invertible. This defines the {\it cocycle KS matrix} $\ast K$, which is invertible with entries $(\ast K)_{\mu\mu'} = \braket{c_\mu}{c_\mu'}$ describing the mutual superpositions of cocycles.

The eigenvalues of $\ast K$ and $K$ are real positive. We further have that

\begin{theorem}\label{more}
The eigenvalues KS matrices are not smaller than $1$. Their spectra coincide, but for the multiplicity of eigenvalue $1$.
\end{theorem}

\begin{proof}
Both facts follow as a special case of a result by Lewkowicz \cite{lewko} that complementary oblique projections have the same singular values but for the multiplicity of $0$ and $1$. We provide a self-contained proof in Appendix \ref{appendix}.
\end{proof}

In particular, if $\C \geq \V-1$  then the cycle KS matrix has eigenvalue $1$ with multiplicity not smaller than $\C - \V+1$ and eigenvalue $\lambda >1$ with the same multiplicity as it occurs in the cocycle KS matrix. Else, if $\C \leq \V-1$, the cocycle KS matrix has eigenvalue $1$ with multiplicity not smaller than $\V-1 - \C$ and eigenvalue $\lambda >1$ with the same multiplicity as it occurs in the cycle KS matrix. Notice that for graphs whose number of edges is exactly $\E = 2 \V - 2$, the two KS matrices have the same spectrum, and since they are diagonalizable they are similar.

Theorem \ref{more} can be interpreted in the light of the polar decomposition $
P = \sqrt{K} \, U$, $Q = V \sqrt{\ast K}$, where $U$ and $V$ are unitary transformations.
The polar decomposition represents any operator as a unique dilation followed by a rotation. The dilation expands the Hilbert space along the eigenvectors corresponding to dilating eigenvalues greater than one. Hence, by the above result complementary projections produce the same expansion, though in different directions.

Eigenvectors of KS matrices and of their corresponding projections are also related. Let $\underline{\ket{w}}$ and  $\overline{\ket{w}}$ be respectively the first $\V-1$ entries the last $\C$ entries of a vector (i.e. respectively the cochord part and the chord part).

\begin{theorem}\label{eigenvectors}
Let $\ket{w}$ be an eigenvector of $P^\dagger P$ corresponding to eigenvalue $\lambda > 1$. Then $\overline{\ket{w}}$ is an eigenvector of $K$ and $\underline{P \ket{w}}$ is a eigenvector of $\ast K$ corresponding to the same eigenvalue $\lambda$.
\end{theorem}
\begin{proof} The first implication is trivial: If $\ket{w}$ is an eigenvector of $P^\dagger P$ then it belongs to the chord space $\mathcal{E}$ (its first $\V-1$ entries vanish); then  $\underline{P^\dagger P \ket{w}} = K \underline{\ket{w}}$. 
As regards the second, in the development of Th.\,\ref{theig} in Appendix \ref{appendix} it is shown that if $\ket{w}$ is an eigenvector of $P^\dagger P$ corresponding to eigenvalue $\lambda > 1$, then $\ket{w''} = Q^\dagger P \ket{w} $ is a nonvanishing eigenvector of $Q^\dagger Q$ corresponding to $\lambda$. It follows that $\ket{w'} = Q Q^\dagger P \ket{w}$ is an eigenvector of $QQ^\dagger$ corresponding to $\lambda$. Then by the same arguments as above ${\ast K} \underline{ P \ket{w}}$ is an eigenvector of $\ast K$. Finally, since ${\ast K}$ is invertible ${\ast K} \underline{ P \ket{w}}= \lambda \, {\ast K}^{-1} {\ast K} \underline{ P \ket{w}}$ and the conclusion follows.
\end{proof}

The eigenvectors of $K$ and/or of $\ast K$ corresponding to eigenvalue $1$ are related to the eigenvectors of the corresponding projections, because the restrictions of the projections to such eigenspaces happen to be orthogonal.
\begin{theorem}\label{tht}
Let $\ket{w}$ be an eigenvector of $P^\dagger P$ corresponding to eigenvalue $1$. Then $\ket{w}$ is an eigenvector of $P$ and of $P^\dagger$ and $\overline{\ket{w}}$ is an eigenvector of $K$ corresponding to eigenvalue $1$.
\end{theorem}
\begin{proof}
This follows from Th.\,\ref{th3} in the Appendix \ref{appendix}.
\end{proof}

KS matrices are related among themselves and to the superposition matrix by the following two theorems.
\begin{theorem}\label{inverse}
Inverse KS matrices are related by the identities
\begin{subequations}\label{eq:inverseth}
\bea
{\ast K}^{-1} & = & 1_{\V-1} -\omega K^{-1} \omega^\dagger  \\
K^{-1} & = &  1_{\C} - \omega^\dagger {\ast K}^{-1} \omega, 
\ees
where $1_{n}$ is the $n$-dimensional identity matrix.
\end{theorem}
\begin{proof} We introduce the invertible operator $\Lambda := \sum_\mu \ket{e_\mu}\bra{e_\mu} + \sum_\alpha \ket{e_\alpha}\bra{c_\alpha}$ that performs a change of basis from  cocycle/chord to the edge basis, and its dual  $\ast \Lambda :=  \sum_\mu \ket{e_\mu}\bra{c_\mu} + \sum_\alpha \ket{e_\alpha}\bra{e_\alpha}$. We have $\Lambda^{-1} = {\ast\Lambda}^\dagger$, and the transpose identity $\Lambda^{-\dagger} = \ast\Lambda$. Combining these two expressions we obtain
\be
{\ast\Lambda} {\ast\Lambda}^\dagger = (\Lambda \Lambda^\dagger)^{-1}.
\ee
Expressing the latter identity in the edge basis, we obtain
\be
 \left(\ba{cc} \ast K & \omega \\ \omega^\dagger & 1_{\C} \ea \right) = \left(\ba{cc} 1_{\V-1} & - \omega^\dagger \\ -\omega & K \ea \right)^{-1}.
\ee
The conclusion follows from the well-known general expression for the inverse of partitioned block matrices in terms of Schur complements \cite{block}.
\end{proof}

\begin{theorem}\label{sup}
KS matrices can be expressed in terms of the superposition matrix by the following identities:
\bes
{\ast K} & = & 1_{\V-1} + \omega \omega^\dagger \\
K & = & 1_{\C} +  \omega^\dagger \omega
\ees
\end{theorem}
\begin{proof}
Matrix $\omega\omega^\dagger $ has entries
\be
(\omega\omega^\dagger)_{\mu\mu'} = \sum_\alpha \left\{\ba{ll}
+1, & \mathrm{if\,} e_\alpha \mathrm{\,belongs\, to\,both\,} c_\mu, c_{\mu'} \mathrm{\,or\,} -c_\mu, -c_{\mu'} \\
-1, & \mathrm{if\,} e_\alpha \mathrm{\,belongs\, to\,both\,} c_\mu, -c_{\mu'} \mathrm{\,or\,} -c_\mu, +c_{\mu'} \\
0, & \mathrm{otherwise}. \ea \right. 
\ee
Off-diagonal entries count the net number of chords in common between two cocycles, as does ${\ast K}$. Diagonal entries count the number of chords in a cocycle; the only difference with ${\ast K}$ is that the latter counts one more for the generating cochord.

Alternatively, we can obtain the same result in an algebraic fashion by plugging the second of Eqs.\,(\ref{eq:inverseth}) into the first and expanding iteratively:
\be
{\ast K}^{-1} = 1_{\V-1} -\omega \omega^\dagger + \omega \omega^\dagger {\ast K}^{-1} \omega \omega^\dagger = \sum_{n=0}^{\infty}\left(-\omega \omega^\dagger \right)^n .
\ee
We can then employ the geometric series and conclude.
\end{proof}

The latter theorem can be condensed into the elegant formula
\be
P^\dagger P + QQ^\dagger = \Id - \Omega^2 = (\Id + \Omega)^\dagger (\Id + \Omega), \label{eq:dirac}
\ee
that establishes a fundamental relation between the several operators that we have introduced in the description. We cannot resist suggesting the analogy of this expression with the representation of the Hamiltonian of the quantum harmonic oscillator in terms of creation and annihilation operators, and to highlight the interplay between the symmetric form $P^\dagger P + QQ^\dagger$ and the antisymmetric form $\Omega$ that is reminiscent of the geometric formulation of Quantum Mechanics \cite{ashtekar}. 

To conclude, a remark on matrix-tree theorems. Notoriously, the determinant of $K$ and $\ast K$ give the number of spanning trees of the graph \cite[Th. 3-10]{nakanishi},
\be
\det K = \det \ast K = \# (\textrm{spanning trees})
\ee
independently of the spanning tree that has been chosen to generate a basis of cycles/cocycles. The choice of a different spanning tree identifies a different basis of cycles, resulting in a transformation $K \to S K S^\dagger$, with $S$ in the special linear group $\mathrm{SL}(\C,\mathbb{R})$ whose generators are the transvection matrices $S_{ij}$ with all $1$'s on the diagonal and all other entries null but for a $+1$ in position $(i,j)$ \cite{conder}. Transvections are not unitary transformations, hence a change of spanning tree generally modifies the spectrum of KS matrices.

\subsection{\label{res4}Duality}

For any embedding of a planar graph $G$ into the surface of a sphere, one can associate a dual graph ${\ast G}$ whose dual covertices are the facets (areas separated by edges) and whose dual coedges connect neighbouring facets along all edges that separate them \cite{biggs}. A cotree ${\ast T} = E \setminus T$ is a spanning tree of ${\ast G}$. Accordingly, a dual decomposition of the identity in terms of dual cycle and cocycle projections is induced by the choice of a tree.  It is well known that duality maps chords into cochords and cycles into cocycles. Then, upon duality,
\be
{\ast P} = Q^\dagger , \quad {\ast Q} = P^\dagger. \label{eq:duality}
\ee
Moreover, it follows that all asterisked operators defined above (e.g. ${\ast\Omega}$, ${\ast K}$ etc.) are indeed the corresponding operators in the dual graph. In particular, the 2-form $\Omega$ is self-dual.

Duality can be extended to nonplanar graphs in several ways. Duality between cycles and cocycles is best encoded in the concept of oriented matroid. In this case all identities above generalize in a formal way, that is not visualizable in terms of graphs. As a second possibiliy, for any finite graph there exists a Riemann surface of high-enough genus where the graph can be embedded. In this case, though, dualization  has to account for topological cycles in the first homology group of the surface that pose an obstruction to the identification of the dual cocycle space with the cycle space. We do not further consider this issue here.

\subsection{\label{res5}Connection to Laplacians}

Spectral analysis on graphs is usually concerned with the spectrum of the Laplacian $\Delta := \partial \partial^\dagger$. In this section we show that this problem can be formulated in terms of a KS spectrum. For simplicity we restrict to connected graphs with no double edges between two vertices and no loops. Then
\be
 \Delta =  \left\{\ba{ll}  \deg_G(v),  & \textrm{if } v = v' \\
-1, & \textrm{if } v' \to v \textrm{ or } v \to v' \\ 0, & \textrm{else}
 \ea
\right.  .
\ee
where $\deg_G(v)$ is the degree of vertex $v$ in $G$. The Laplacian has positive spectrum, with eigenvalue $\lambda = 0$ once degenerate.

Now let $\tilde{G}$ be the directed graph obtained by adding one vertex $v_0$, and edges connecting $v_0$ to any other vertex in $V$, with orientation out of $v_0$. We have $|\tilde{E}| = \E + \V$, $|\tilde{V}| = \V+1$, $|\tilde{C}| = \E$. For the working example considered in Appendix \ref{example}, we have
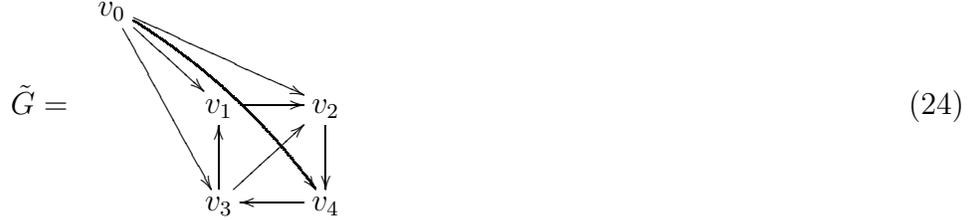
\be
\tilde{G} =  \ba{c}\xymatrix{ v_0     \ar@{->}[dr]  \ar@/^/[ddrr]     \ar@{->}[drr]      \ar@{->}[ddr]    \\  & v_1 \ar@{->}[r] & v_2  \ar@{->}[d] \\ & v_3 \ar@{->}[u]   \ar@{->}[ur] & v_4 \ar@{->}[l] }  \ea
\ee
To generate a basis of cocycles, let us choose the spanning tree that emanates out of vertex $v_0$
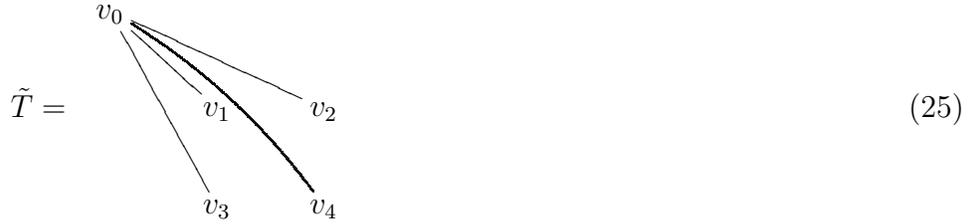
\be
\tilde{T} =  \ba{c}\xymatrix{ v_0     \ar@{-}[dr]  \ar @{-} @/^/[ddrr]     \ar@{-}[drr]      \ar@{-}[ddr]    \\  & v_1   & v_2  \\ & v_3 & v_4 }  \ea
\ee
yielding a basis of cocycles
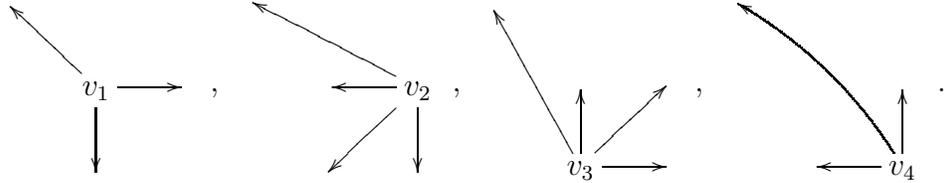
\be
\ba{c}\xymatrix{\ar@{<-}[dr]  \\  & v_1 \ar@{->}[r] &   \\ & \ar@{<-}[u] &}  \ea,
\ba{c}\xymatrix{ \ar@{<-}[drr] \\  &  \ar@{<-}[r] & v_2  \ar@{->}[d] \\ &  \ar@{<-}[ur] & }  \ea ,
 \ba{c}\xymatrix{ \ar@{<-}[ddr]    \\  &  & \\ & v_3 \ar@{->}[u]   \ar@{->}[ur] & \ar@{<-}[l]}  \ea ,
  \ba{c}\xymatrix{ \\  & & \ar@{<-}[d] \\ & & v_4 \ar@{->}[l] \ar@/_/[uull]    }  \ea . \nonumber
\ee

Let ${\ast\tilde{K}}$ be the cocycle KS matrix of graph $\tilde{G}$ with respect to the basis of cocycles generated by $\tilde{T}$. It is straightforward that
\be
\tilde{K} = \Delta + 1_{|V|}.
\ee
In particular, the spectrum of $\tilde{K}$ is the Laplacian spectrum shifted up by $1$. So, it is always possible to map the spectrum and eigenfunctions of the Laplacian into those of a cocycle KS matrix. 

\section{\label{thermo}Application to thermodynamic networks}

A thermodynamic network \cite{oster} might be defined as a directed graph with two variables associated to each edge, the currents $\ket{j}$ and their conjugate forces $\bra{f}$, e.g. electric currents and voltage drops. We refer to these as the {\it mesoscopic} observables. The {\it dissipated power} (or {\it entropy production rate}) is the bilinear form
\be
\sigma = \braket{f}{j}.
\ee
We will now change basis from the edge basis to the cycle/cochord basis for vectors, and cocycle/chod basis for covectors, by means of the transformation matrices defined in the proof of Th.\,\ref{eq:inverseth}:
\bea
\Lambda\ket{j} & = & \sum_\mu J_\mu \ket{e_\mu} + \sum_\alpha J_\alpha \ket{e_\alpha} \\
{\ast\Lambda} \ket{f} & = & \sum_\mu F_\mu \ket{e_\mu} + \sum_\alpha F_\alpha \ket{e_\alpha}
\eea
Here $J_\alpha = \braket{e_\alpha}{j}$ is the a {\it vortex current} flowing along  a chord, $J_\mu = \braket{c_\mu}{j}$ is the {\it tidal current} flowing from a source set to the sink set of a cocycle, $F_\mu = \braket{e_\mu}{f}$ is the potential drop along a cochord and $F_\alpha = \braket{c_\alpha}{f}$ is the circuitation of the forces along a cycle. In the context of Markov jump processes, the latter have been interpreted by the Author in Ref.~\cite{gauge} as Wilson loops of an abelian gauge connection arising from an invariance of the theory under a change of Bayesian prior probability.

These are the {\it macroscopic observables} that characterize a thermodynamic network. Importantly, $F_\alpha$ and $J_\mu$ are nonlocal. Nevertheless, letting $w_i = \braket{e_i}{w}$ and defining the Poisson bracket
\be
\{A,B\} = \sum_i \left( \frac{\partial A}{\partial j_i}  \frac{\partial B}{\partial f_i} - \frac{\partial B}{\partial f_i}  \frac{\partial A}{\partial j_i} \right)
\ee
the macroscopic observables are canonically conjugate
\be
\{J_i,F_{i'}\} = \delta_{i,i'}. 
\ee
We can express the mesoscopic observables in terms of macroscopic ones by decomposing via the projections
\bes
\ket{j} & = & P\ket{j} + Q \ket{j} = \sum_\mu J_\mu \ket{e_\mu} + \sum_\alpha J_\alpha \ket{c_\alpha}  \\ 
\ket{f} & = & P^\dagger \ket{f} + Q^\dagger \ket{f} = \sum_\mu F_\mu \ket{c_\mu} + \sum_{\alpha} F_\alpha \ket{e_\alpha}.
\ees
The first identity decomposes the current field in terms of ``tides'' and of ``vortices''. The second decomposes the force field in terms of a conservative part and a curl, in a way that is analogous to the Hodge-Helmholtz decomposition of vector fields. 

Kirchhoff's Current Law is satisfied when
\be
Q \ket{j} = 0, \quad  \implies J_\mu = 0. 
\ee
Kirchhoff's Loop Law is satisfied when
\be
P^\dagger \ket{f}  = 0, \quad  \implies F_\alpha = 0.
\ee
When both hold, one is at {\it equilibrium}. When either one is violated, the system is in a {\it nonequilibrium} state that is either driven by nonconservative external forces, or by injected external currents. For an electrical network the latter are electromotive forces, the former are current generators. For a chemical reaction network \cite{polespo}, the latter are fixed chemical potentials of certain substrates regulated by membranes or pores to the environment ({\it in vivo} conditions), the former are initially disproportioned concentrations of chemicals ({\it in vitro} conditions). In realistic nonequilibrium states both sources of dissipation contribute ({\it in situ} conditions). 

The entropy production rate can be expressed in terms of the macroscopic observables as
\be
\sigma = \braket{f}{j} = \bra{f} {\ast\Lambda}^\dagger \Lambda \ket{j} = \sum_\alpha F_\alpha J_\alpha + \sum_\mu F_\mu J_\mu.
\ee
The entropy production rate is thus decomposed in a term that only vanishes if Kirchhoff's Current Law holds, and a term that only vanishes if Kirchhoff's Loop Law holds. At equilibrium the entropy production rate vanishes.

A situation of interest is the {\it linear regime} where forces are linearly related to currents. We consider here the situation when $\ket{f} = \ket{j}$, which in an electrical network corresponds to all unit resistances. Observing from Eq.\,(\ref{eq:ompq}) that $I - \Omega = Q + P^\dagger$ and given Eq.\,(\ref{eq:dirac}), we obtain the linear regime expression for the entropy production rate
\bea
\sigma \;=\; \braket{j}{j} & = & \bra{j} (Q^\dagger+P) (I - \Omega^2)^{-1} (Q+P^\dagger) \ket{j} \nonumber \\
& = & \sum_{\mu,\mu'} ({\ast K}^{-1})_{\mu\mu'} J_\mu  J_{\mu'} +  \sum_{\alpha,\alpha'} (K^{-1})_{\alpha \alpha'} F_\alpha  F_{\alpha'}.\label{eq:linregpro}
\eea
This interesting formula shows that when one expresses the local quadratic form $\braket{j}{j}$ in terms of the macroscopic observables, inverse KS matrices appear. In the context of nonequilibrium thermodynamics, inverse KS matrices can be interpreted as Onsager linear response matrices. Then, the isospectral properties of KS matrices partake to the paradigm of fluctuation-dissipation that establishes a relationship between spontaneous relaxation, and the response after a perturbation. Finally, notice that the two operators appearing in Eq.\,(\ref{eq:linregpro})
\be
\hhat{P}' = \sum_{\alpha,\alpha'}  \ket{c_\alpha}K^{-1}_{\alpha,\alpha'}\bra{c_{\alpha'}}, \qquad \hhat{Q}' = \sum_{\mu,\mu'} \ket{c_\mu} {\ast K}^{-1}_{\mu\mu'} \bra{c_{\mu'}}
\ee
are the complementary orthogonal projections associated with the decomposition of the edge space into the cycle and the cocycle spaces, $\mathbb{R}^e = \mathcal{C} \oplus \mathcal{C}_\ast$.

\section{\label{conclu}Conclusions and perspectives}

In this paper we provided some elements of an algebraic theory of oblique complementary projections associated to the decomposition of a graph's edge space into cycles and cocycles. The theory pivots on the choice of a spanning tree and the ensuing choice of basis for the cycle and the cocycle vector spaces. The formalism allows to prove some novel results regarding so-called Kirchhoff-Symanzik matrices. An illustrative application of the theory to thermodynamic networks allows to appreciate the consistency of the results against physical intuition of flows and forces on networks. As briefly discussed by the Author in Ref.\,\cite{duality}, many but not all of the results above can be generalized to weighted graphs that more often occur in physical applications, such as the linear regime of thermodynamic networks and Feynman graphs. Such generalization corresponds to the introduction of a non-Euclidean metric.

The application of more advanced tools from projection algebra (the resolvent formalism, singular value decompositions etc.) might give further insights on the topology of graphs. The Author believes that the 2-form $\Omega$ is crucial in capturing certain topological features of a graph. It would be interesting for example to study its behavior under deletion and contraction. Furthermore, since all results in this paper regard the homology of graphs, the construction could be generalized to manifolds of higher dimensions.

\subsubsection*{Acknowledgements}

The Author is grateful to D. Mugnolo who has stirred his interest for graph theory over the years, and to professor D. B. Szyld for discussion and suggestions on projection algebra. The research was supported by the National Research Fund of Luxembourg in the frame of Project
No. FNR/A11/02 and the AFR Postdoc Grant 5856127.

\appendix 

\section{\label{appendix}Singular values of complementary oblique projections}

Lewkowicz \cite{lewko} proved that complementary oblique projections have the same singular values but for the multiplicities of eigenvalues $1$ and $0$, and that all positive eigenvalues are not smaller than $1$, generalizing the better-known norm identity $\|P\|=\|Q\|$ \cite{szyld}. Below we provide a simple proof.

Consider an $n$-dimensional Hilbert space $\mathcal{H}$. Let $P:\mathcal{H} \to \mathcal{H}$ be a projection, $P^2 = P$, neither null nor the identity, and $Q = I-P$ be its complementary projection. The adjoints $P^\dagger$ and $Q^\dagger$ are also complementary projections. In general, such projections are oblique, i.e. not orthogonal, $P^\dagger \neq P$. As per basic linear algebra (see e.g. Ref.~\cite{horn}), being Gramian $P^\dagger P$ is positive semidefinite, its kernel is the kernel of $P$, it admits a real nonnegative spectrum, and it is  nondefective (i.e. it can be diagonalized by a unitary transformation). Moreover, $PP^\dagger $ and $P^\dagger P$ have the same spectrum.

\begin{theorem}\label{th2} The nonnull eigenvalues of $P^\dagger P$ are not smaller than $1$. 
\end{theorem}

\begin{proof}
Let $P^\dagger P \ket{w} = \lambda \ket{w}$, $\lambda > 0$. Taking the scalar product by  $\ket{w}$ yields
\bea
\|P\ket{w}\|^2 \; = \;  \lambda \braket{w}{w}
& = & \lambda \bra{w}(P^\dagger + Q^\dagger)(P+Q)\ket{w} \nonumber \\
& = & 
\lambda  \big( \bra{w}P^\dagger P\ket{w} +
 \bra{w}Q^\dagger Q\ket{w} + 2 \Re \bra{w} Q^\dagger P \ket{w} \big)  \nonumber \\
& = & \lambda (  \|P\ket{w}\|^2 -
 \|Q\ket{w}\|^2+  2 \Re \bra{w} Q^\dagger \ket{w} \big),
\eea
where in the latter passage we plugged $P = I-Q$ into the last term between parenthesis.  The eigenvalue equation also tells us that $Q^\dagger \ket{w} = 0$. Then
\be
(\lambda -1) \|P\ket{w}\|^2 = \lambda \|Q\ket{w}\|^2.
\ee
Since $\lambda >0$ and both norms are nonnegative, then $\lambda \geq 1$.
\end{proof}

\begin{theorem}\label{th3}
If $\ket{w}$ is an eigenvector of $P^\dagger P$ corresponding to eigenvalue $1$, then $\ket{w} = P\ket{w} = P^\dagger \ket{w}$, and vice versa.
\end{theorem}

\begin{proof}
When $\lambda = 1$ we have $\|Q\ket{w}\| = 0$, hence $Q\ket{w} = 0$, hence $\ket{w} = P\ket{w}$.  Again, the eigenvector equation implies that $\ket{w} = P^\dagger P \ket{w} = P^\dagger \ket{w}$. Then $\ket{w} = P\ket{w} = P^\dagger \ket{w}$. The converse is trivial.  
\end{proof}

The interpretation of this latter result is that $P$ behaves as an orthogonal projection when restricted to the eigenspace $\mathcal{H}_1$ of $PP^\dagger$ corresponding to eigenvalue $1$. Notice that, although $P^\dagger P \ket{w} = 0$ implies $P \ket{w} = 0$, in general $P^\dagger \ket{w} \neq 0$, so that $P$ restricted to the null eigenspace $\mathcal{H}_0$ is not orthogonal.

We can now move to the main result.

\begin{theorem}\label{theig}All eigenvalues of $P^\dagger P$ greater than $1$ are eigenvalues of $Q^\dagger Q$, with the same multiplicity.
\end{theorem}
\begin{proof}
Let $\ket{w}$ be an eigenvector of $P^\dagger P$ corresponding to $\lambda \neq 0$. Then $\ket{w'} = P\ket{w}$ is an eigenvector of $PP^\dagger$ corresponding to $\lambda$. Expanding $QQ^\dagger = I - P - P^\dagger + PP^\dagger$ and noting that  $(P + P^\dagger) \ket{w'} = (P + P^\dagger) P\ket{w} = 
\ket{w'} + P^\dagger P \ket{w'}$, we obtain
$$
QQ^\dagger \ket{w'} = \lambda \ket{w'} - P^\dagger P\ket{w'}.
$$
Applying $Q^\dagger$ we get
$$
Q^\dagger QQ^\dagger \ket{w'} = \lambda Q^\dagger \ket{w'}. 
$$
Then, if vector $\ket{w''} = Q^\dagger P \ket{w}$ is nonvanishing, then it is an eigenvector of $Q^\dagger Q$ corresponding to $\lambda$.
We have to check that $\ket{w''}$ is nonnull. In the case $\lambda = 1$, by Th.\;\ref{th3} we have $P\ket{w} = P^\dagger \ket{w}$, and $\ket{w''} = Q^\dagger P^\dagger \ket{w}= 0$. Therefore, one cannot conclude that $\lambda=1$ is an eigenvalue of $Q^\dagger P$. As for $\lambda > 1$, we decompose
\be
P\ket{w} = \left(P^\dagger + Q^\dagger\right) P\ket{w} = P^\dagger P\ket{w} + \ket{w''} = \lambda \ket{w} + \ket{w''}. 
\ee
If {\it ad absurdum} $\ket{w''} = 0$, then applying $P$ one finds the contradiction $\lambda = 1$.
To summarize, for any eigenvector $\ket{w}$ of $P^\dagger P$ corresponding to eigenvalue $\lambda > 1$, vector $\ket{w''} = Q^\dagger P \ket{w}$ is an eigenvector of $Q^\dagger Q$ corresponding to eigenvalue $\lambda$.
Finally, let's discuss multiplicities. Let $\ket{w_1} \neq \ket{w_2}$ be two eigenvectors of $P^\dagger P$ corresponding to $\lambda > 1$. Then $\ket{w_1} - \ket{w_2}$ is also an eigenvector of $P^\dagger P$ corresponding to $\lambda$. Let $\ket{w''_1} = Q^\dagger P \ket{w_1}$ and  $\ket{w''_2} = Q^\dagger P \ket{w_2}$ be eigenvectors of $Q^\dagger Q$ corresponding to $\lambda$. If {\it ad absurdum} $\ket{w''_1} = \ket{w''_2}$, then $Q^\dagger P (\ket{w_1} - \ket{w_2}) = 0$. As discussed above, this implies that $\ket{w_1} - \ket{w_2}$ is an eigenvector of $P^\dagger P$ corresponding to $\lambda = 1$, which violates the hypothesis. \end{proof}

\begin{theorem}Let $r_1$ and $r_0$ be the multiplicities of eigenvalues $1$ and $0$ in the spectrum of $P^\dagger P$. Then the spectrum of $Q^\dagger Q$ contains eigenvalue $0$ with multiplicity $n-r_0$ and eigenvalue $1$ with multiplicity $2r_0 + r_1-n$.
\end{theorem}
\begin{proof}
The ranks of $P^\dagger P$ and $Q^\dagger Q$ equal respectively that of $P$ and $Q$. Then the rank $n- r_0$ of $P^\dagger P$ is the nullity of $Q^\dagger Q$. Since $r_0 + r_1 + \sum_{\lambda>1} 1 = n$, by simple balances one concludes. 
\end{proof}

\section{\label{example}Example}

Consider the following oriented graph and its incidence matrix:
\be G =  \ba{c}\xymatrix{ v_1 \ar@{->}[r]^{e_1} & v_2  \ar@{->}[d]^{e_4} \\ v_3 \ar@{->}[u]^{e_5}  \ar@{->}[ur]^{e_2} & v_4 \ar@{->}[l]^{e_3} }  \ea , \qquad \partial = \left(\ba{ccccc}
-1 & 0 & 0 & 1 & 0 \\
1 & -1 & 0 & 0 & 1 \\
0 & 1 & -1 & 0 & 0 \\
0 & 0 & 1 & -1 & -1
\ea \right). \label{eq:embed}
\ee
As spanning tree we choose
		\be 
		\xymatrix{ v_1 \ar@{-}[r] & v_2 
		\\  v_3 \ar@{-}[ur] & v_4 \ar@{-}[l] }
		\ee
Two independent cycles are generated by adding chords to the spanning tree, identifying a cycle and orienting it in the direction of the generating chord:
		\bes
		\ba{c}\xymatrix{ \ar@{-}[r] & \\   \ar@{-}[ur] & \ar@{-}[l] }\ea
		+ 
		\ba{c}\xymatrix{ \ar@{.}[r]   &  \ar@{->}[d] \\   \ar@{.}[ur] \ar@{.}[u] & \ar@{.}[l]    }\ea
		& \to &
		\ba{c}\xymatrix{&   \ar@{->}[d] \\   \ar@{->}[ur]  & \ar@{->}[l] }\ea
		\\
		\ba{c}\xymatrix{ \ar@{-}[r] & \\   \ar@{-}[ur] & \ar@{-}[l] }\ea
		+ 
		\ba{c}\xymatrix{ \ar@{.}[r]  \ar@{<-}[d] & \\   \ar@{.}[ur] & \ar@{.}[l]  \ar@{.}[u]  }\ea
		& \to & 
		\ba{c}\xymatrix{ \ar@{->}[r]  \ar@{<-}[d] & \\   \ar@{<-}[ur] & }\ea
		\ees
Three independent cocycles are generated by subtracting cochords, thus disconnecting the vertex set into sources (circles) and sinks (disks) and identifying the unique cocycle out of the source towards the sink,
		\bes
		\ba{c}\xymatrix{ \ar@{-}[r] & 
		\\   \ar@{-}[ur] & \ar@{-}[l] }\ea -
		\ba{c}\xymatrix{ \ar@{->}[r]  \ar@{.}[d] & 
		\\   \ar@{.}[ur] & \ar@{.}[l]  \ar@{.}[u]  }\ea
		& \to & 		\ba{c}\xymatrix{ \circ \ar@{->}[r]  \ar@{->}[d] & \bullet
		\\ \bullet   \ar@{.}[ur]   & \bullet    \ar@{.}[l]  }\ea \\
		\ba{c}\xymatrix{ \ar@{-}[r] & 
		\\   \ar@{-}[ur] & \ar@{-}[l] }\ea  -
		\ba{c}\xymatrix{ \ar@{.}[r]  \ar@{.}[d] & 
		\\   \ar@{->}[ur] & \ar@{.}[l]  \ar@{.}[u]  }\ea
		& \to & 	
		\ba{c}\xymatrix{\bullet \ar@{<-}[d]  \ar@{.}[r] & \bullet
		\\  \circ \ar@{->}[ur] \ar@{.}[r] & \circ \ar@{->}[u] }\ea \\
				\ba{c}\xymatrix{ \ar@{-}[r] & 
		\\   \ar@{-}[ur] & \ar@{-}[l] }\ea -
		\ba{c}\xymatrix{ \ar@{.}[r]   &  \ar@{.}[d] \\   \ar@{.}[ur] \ar@{.}[u] & \ar@{->}[l]    }\ea
		& \to &
		\ba{c}\xymatrix{\bullet \ar@{.}[r] & \bullet
		\\ \bullet  \ar@{.}[ur] & \ar@{->}[l] \ar@{->}[u] \circ }\ea
		\ees
The covectors associated to chords, cochords, cycles and cocycles read
\bea
\bra{e_1} = \left(\ba{ccccc} 1 & 0 & 0 & 0 & 0 \ea\right) &\quad & \bra{c_1} =  \left(\ba{ccccc} 1 & 0 & 0 & 0 & -1 \ea\right)\nonumber \\
\bra{e_2} = \left(\ba{ccccc} 0 & 1 & 0 & 0 & 0 \ea\right) &\quad & \bra{c_2} =  \left(\ba{ccccc} 0 & 1 & 0 & -1 & 1 \ea\right)\nonumber \\
\bra{e_3} = \left(\ba{ccccc} 0 & 0 & 1 & 0 & 0 \ea\right) &\quad & \bra{c_3} =  \left(\ba{ccccc} 0 & 0 & 1 & -1 & 0 \ea\right)\nonumber \\
\bra{e_4} = \left(\ba{ccccc} 0 & 0 & 0 & 1 & 0 \ea\right) &\quad & \bra{c_4} =  \left(\ba{ccccc} 0 &  \hspace{0.08cm}  1 &  \hspace{0.08cm}  1 &  \hspace{0.08cm}  1 & \hspace{0.08cm} 0 \ea\right) \nonumber \\
\bra{e_5} = \left(\ba{ccccc} 0 & 0 & 0 & 0 & 1 \ea\right) &\quad & \bra{c_5} =  \left(\ba{ccccc} 1 & -1 & 0 & 0 & 1 \ea\right).
\eea
Taking the outer products and summing we obtain
\be
P = \left(\ba{ccccc}
0 & 0 & 0 & 0 & 1 \\
0 & 0 & 0 & 1 & -1 \\
0 & 0 & 0 & 1 & 0 \\
0 & 0 & 0 & 1 & 0 \\
0 & 0 & 0 & 0 & 1
\ea\right),
\quad Q =
 \left(\ba{ccccc} 
1 & 0 &  0 & 0 & -1 \\
0 & 1 & 0 & -1 & 1 \\
0 & 0 &  1 & -1 & 0 \\
0 & 0 & 0 & 0 & 0 \\
0 & 0 & 0 & 0 & 0 \ea\right).
\ee
which add up to unity. We further obtain the KS matrices:
\be
K = \left(\ba{cc} 3 & -1 \\ -1 & 3 \ea\right), \quad {\ast K} = \left(\ba{ccc} 2 & -1 & 0 \\ -1 & 3 & 1 \\ 0 & 1 & 2 \ea\right)
\ee
with eigenvalues respectively $(2,4)$ and $(1,2,4)$, verifying Th.\,\ref{more}. Notice that the choice of spanning tree $(e_1,e_3,e_4)$ gives different KS matrices, with a different spectrum. 

The eigenvectors $\ket{\lambda}$, $\lambda > 1$ of $P^\dagger P$ are
\be
\ket{2} = \left(\ba{c}0 \\  0 \\  0 \\ 1 \\ 1 \ea \right), \quad \ket{4} = \left(\ba{c} 0 \\ 0 \\ 0 \\ -1 \\ 1 \ea \right)
\ee
Then by  Th.\,\ref{eigenvectors} the eigenvectors of $K$ are
\be
\overline{\ket{2}} =  \left(\ba{c} 1 \\ 1 \ea \right), \quad \overline{\ket{4}} =  \left(\ba{c} -1 \\ 1 \ea \right)
\ee
and the eigenvectors of ${\ast K}$  are
\be
\underline{P \ket{2}} =  \underline{\left(\ba{c} 1 \\ -2 \\  -1 \\ -1 \\ 1 \ea \right)} =\left(\ba{c} 1 \\ -2 \\  -1 \ea \right), \quad \underline{P \ket{4}} =  \underline{\left(\ba{c} 1 \\ 0 \\  1 \\ 1 \\ 1 \ea \right)} = \left(\ba{c} 1 \\ 0 \\  1 \ea \right)
\ee
as can be immediately verified. As regards eigenvalue $1$ of of $QQ^\dagger$, the corresponding eigenvector is $\bra{1} = (-1,-1,1,0,0)$ and Th.\,\ref{tht} can be verified, $Q \ket{1} = Q^\dagger \ket{1} = \ket{1}$. Finally, we have
\be
\Omega = \left( \ba{ccccc}   0 & 0 & 0 & 0 & 1 \\ 0 & 0 & 0 & 1 & -1 \\ 0 & 0 & 0 & 1 & 0 
\\ 0 & -1& -1 & 0 & 0 \\ -1 & 1 & 0 & 0 & 0 \ea \right), \quad I - \Omega^2 =  \left( \ba{ccccc}  2 & -1  & 0 & 0  & 0 \\ -1 & 3 & 1 & 0 & 0 \\ 0 & 1 & 2 & 0 & 0 
\\ 0 & 0 & 0 & 3 & -1 \\ 0 & 0 & 0 & -1 & 3 \ea \right)
\ee
and Eq.\,(\ref{eq:dirac}) is verified. Finally, the dual graph ${\ast G}$ (with respect to the embedding of $G$ the page as in Eq.\,(\ref{eq:embed})) is
\be
{\ast G} =  \ba{c}
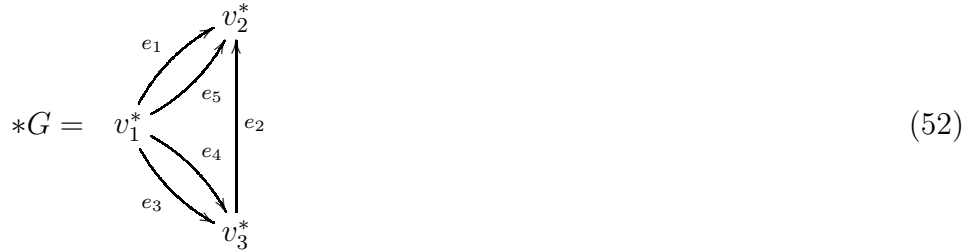
\xymatrix{& v^\ast_2 \\ v^\ast_1 \ar@{->}@/^/[ur]^{e_1} \ar@{->}@/_/[ur]_{e_5} \ar@{->}@/^/[dr]^{e_4} \ar@{->}@/_/[dr]_{e_3}  \\ & v^\ast_3 \ar@{->}[uu]_{e_2}  }  \ea
\ee 
Repeating the above analysis for the dual graph, one can reproduce the duality relationships in Eq.\,(\ref{eq:duality}).

\end{document}